\documentclass{article}

\usepackage{PRIMEarxiv}
\usepackage{mathtools}
\usepackage{amsmath,amsfonts,amsthm,stackrel}
\usepackage[utf8]{inputenc} 
\usepackage[T1]{fontenc}    
\usepackage{hyperref}       
\usepackage{url}            
\usepackage{booktabs}       
\usepackage{amsfonts}       
\usepackage{nicefrac}       
\usepackage{microtype}      
\usepackage{lipsum}
\usepackage{fancyhdr}       
\usepackage{graphicx}       
\graphicspath{{media/}}     
\usepackage{babel}
\usepackage{etoolbox}

\newtheorem{theorem}{Theorem}[section]
\newtheorem{definition}[theorem]{Definition}

\newtheorem{lemma}[theorem]{Lemma}

\AfterEndEnvironment{theorem}{\noindent\ignorespaces}
\AfterEndEnvironment{definition}{\noindent\ignorespaces}
\AfterEndEnvironment{corollary}{\noindent\ignorespaces}
\AfterEndEnvironment{lemma}{\noindent\ignorespaces}
\AfterEndEnvironment{claim}{\noindent\ignorespaces}

\usepackage{float}

\title{Adjusting Queries to Statistical Procedures Under
Differential Privacy}

\author{
  Tomer Shoham \\
  Department of Computer Science and Federmann Center for the Study of Rationality\\
  The Hebrew University of Jerusalem \\
  \texttt{Tomer.Shoham@mail.huji.ac.il} \\
   \And
  Yosef Rinott \\
  Department of Statistics and Federmann Center for the Study of Rationality\\
  The Hebrew University of Jerusalem \\
  \texttt{Yosef.Rinott@mail.huji.ac.il} \\
}

\begin{document}
\maketitle

\begin{abstract}
  We consider a dataset $S$ held by an agency, and a vector query of  interest, $f(S) \in \mathbb{R}^k$, to be posed by an analyst, which contains the information required for certain planned statistical inference. The agency releases the requested vector query with noise that guarantees a given level of Differential Privacy --   DP$(\varepsilon,\delta)$ --   using  the well-known Gaussian mechanism. The analyst can choose to pose the vector query $f(S)$ or to adjust it by a suitable transformation that can make the agency's response more informative. For any given level of privacy DP$(\varepsilon,\delta)$ decided by the agency, we study natural  situations where the analyst can achieve better statistical inference by adjusting the query with a suitable simple explicit transformation.
\end{abstract}


\section{Introduction}
Throughout the paper we consider a dataset or sample $S$ given as an $n \times d$ matrix, where each row pertains to an individual in the sample, and $d$ variables are measured for each of the $n$ participant in the sample. The sample $S$ is held by some agency and an analyst is interested in a vector function $f(S) =(f_1(S),\ldots,f_k(S)) \in \mathbb{R}^k$  of the data, to be called a query. Thus, a query consists of $k$ functions of the data to be posed to the agency by the analyst. We consider throughout the case $k>1$.  We assume that for privacy considerations, the agency releases the response to the query $f(S)$ with noise, using a standard Gaussian mechanism that adds independent $N(0,\sigma^2)$ noise to each coordinate of $f(S)$. The distribution of the added noise is always assumed to be known to the analyst, a standard assumption in the differential privacy literature. Two samples $S$ and $S'$ are said to be neighbors, denoted by $S \sim S'$, if they differ by a  single individual, i.e., a single row. See, e.g., \cite{Dwork_2014_book} for all needed details on Differential Privacy (henceforth DP). When we consider $S$ and $S'$ together we always assume that they are neighbors.

More generally, consider a noise mechanisms applied to $S$ via a query $h(S) \in \mathbb{R}^k$ of the form ${\cal M}_h(S)= h(S)+U \in \mathbb{R}^k$, where $U$ is a random vector. A mechanism ${\cal M}_h$ is said to be DP$(\varepsilon,\delta)$ if for all (measurable) sets $E$ we have 
\begin{equation}\label{eq:DPDP}
P({\cal M}_h(S) \in E) \le e^\varepsilon P({\cal M}_h(S') \in E)+\delta\end{equation}
for all $S \sim S'\in 
{\mathfrak D}$,  where the probability refers to the randomness of $U$, and ${\mathfrak D}$ is the universe of potential datasets. For example, if $S$ is a sample of a given size $n$ from some population, then ${\mathfrak D}$ is the universe of all samples that could have been drawn and considered for dissemination. The standard definition of DP takes ${\mathfrak D}$ to be a product $C^n$ where $C$ consists of all possible rows. 
Our results hold for any given $\varepsilon>0$ and $\delta \in (0,1)$,  which we fix for the rest of this paper.

Our goal is to describe some simple natural examples where posing a linear transformation of the query $f(S)$, getting the agency's response via a mechanism that guarantees DP, and inverting the response to obtain the required information on $f(S)$ yields better inference on $f(S)$ in the case where $S$ is a given fixed dataset, and on the model that generates $f(S)$, when $S$ is a random sample.

\textbf{Some related work}: The principle of modifying queries for better results is not new. The Matrix Mechanism  (MM) is put forward in a line of work that started with \cite{OptHisQue}. Further literature includes  \cite{edmonds2020power, MM, li2015matrix, HDMM} and numerous references therein. For given queries, MM
linearly modifies the original data by applying a matrix that depends on the queries to be answered. The modified data is released with noise, and the
answer to the original queries is computed. The above literature studies algorithms for finding optimal modifying matrices that minimizes the distance between the original queries and the mechanism's output relative to different utility metrics. 

The difference between the above papers and ours is fourfold. First, we provide simple explicit transformations for the situations we consider rather than a numerical optimization algorithm; second, we consider continuous data while  MM is directed mostly toward frequency tables and counting queries. Applying MM to continuous data leads to large and sparse tables of counts (contingency tables) and high complexity of the algorithm; third, our transformations are aimed toward specific statistical goals, rather than standard norms (metrics); and fourth, we consider also random datasets where inference is on the data-generating process.

 For certain specific utility metrics and queries,  optimal noise mechanisms have been  found; see, e.g., \cite{ghosh2012universally, gupte2010universally}, but in general many researchers consider simple mechanisms with a well-known distribution (e.g., addition of Laplace or Gaussian iid noise) without considering optimality. 
In \cite{R2DP} a noise mechanism is proposed where the variance of the Laplace noise is random. Given a query and its sensitivity, an algorithm for optimal choice of the distribution of the Laplace variance from a certain class of distributions is provided. Optimality is with respect to the expected distance (metric) between the original query and the output of the mechanism. 

When  a dataset is randomly generated by some assumed distribution, it is well known that the analyst  has to adjust the statistical procedure to the distribution of the observed data, taking the distribution of the added noise into account; see, e.g., \cite{solea2014differentially, rogers2017new, wang2018statistical, canonne2019structure, gaboardi2016differentially} and  references therein.  Most of these results are asymptotic. 
   
\section{Fixed (non-random) datasets}\label{sec:fixedNR}
Consider  a dataset $S$ held by an agency and an analyst who poses a query $f(S)$ in terms of measurement units of his choosing. For example, the components of $f(S)$ could be average age, average years of schooling, and median income in the sample $S$. The observed response is given with noise through a privacy mechanism applied by the data-holding agency. The analyst's   goals are to construct a confidence region for $f(S)$ and to test simple hypotheses about it.  For any given level $\varepsilon,\delta$ of DP, we show that instead of posing the query $f(S)$, the analyst can obtain a smaller confidence region for $f(S)$ by  computing it from  a query of the form $f_\xi(S)= Diag(\xi)^{1/2}f(S)$ for a suitable  $\xi \in \mathbb{R}_{\ge 0}^k$ (a vector having nonnegative coordinates), where $Diag(\xi)$ is a diagonal matrix whose diagonal elements form the vector  $\xi$.
For the goal of testing hypotheses, it turns out that a different choice of $\xi$ maximizes the power of the standard likelihood-ratio test.  Thus, the analyst can achieve  better inference by adjusting his queries to the planned statistical procedure.

Consider a row $(x_1,\ldots,x_d)$ in the sample $S$. For simplicity we assume that $x_i \in C_i$ for $i=1,\ldots,d$ for suitable sets $C_i$.
In this case each row is in the Cartesian product $C := {C}_1\times\ldots\times {C}_d$ and we set  
${\mathfrak D}:=C^n$. We can also write ${\mathfrak D} \coloneqq {\mathcal C}_1\times\ldots\times {\mathcal C}_d$, where ${\mathcal C}_i$ is the set of $n$-vectors whose coordinates are all in $C_i$.
We assume that the agency releases data under DP$(\varepsilon,\delta)$ relative to this universe ${\mathfrak D}$, which is known to both the agency and the analyst.

In Section \ref{sec:fixedNR} we assume that the components $f_i$ of the vector query $f= \left(f_1,\ldots,f_k\right)$ are functions of disjoint sets of columns of $S$. This assumption is not needed in Section \ref{sec:randomdata}. The quantity 
$$\Delta(f):=\operatorname*{max}_{S\sim S' \in {\mathfrak D}}||f(S)-f(S')||,$$ where $|| \cdot ||$ denotes the $L_2$ norm, is known as the sensitivity of $f$; higher sensitivity requires more noise for DP.
Under simple assumptions on the functions $f_i$ such as monotonicity, the agency can readily compute $\Delta(f)$, as well as
\begin{equation}\label{eq:tildeS}
(\widetilde{S},\widetilde{S}') := \operatorname*{argmax}_{S\sim S' \in {\mathfrak D}}||f(S)-f(S')||;
\end{equation}
see Lemma \ref{le:ontildeS}, where it is shown that the maximization can be done separately for each coordinate of $f$. In general, the maximum in \eqref{eq:tildeS} is not unique, in which case $\arg\max$ is a set of pairs. 
	
The agency plans to release a response to the query $f(S)$ via a standard Gaussian mechanism; that is, the response is given by $${\cal M}(S) = f(S)+U  \text{  where  } U \sim N(0, \sigma^2I)$$ and  I is the $k \times k$ identity matrix. The variance $\sigma^2$ is the minimal variance such that the mechanism  satisfies DP$(\varepsilon,\delta)$ for given $\varepsilon,\delta$; it can be determined by Lemma \ref{le:BW} below,  which appears in \cite{balle2018improving}. This variance depends on $\Delta(f)$; however, here $f$ is fixed and hence suppressed. 

Consider a family of queries adjusted by $Diag(\xi)$: $$f_\xi(S):=Diag(\xi)^{1\!/2}f(S)=\big({\xi_1}^{\!\!\! 1\!/2}f_{1}(S),
\ldots,{\xi_k}^{\!\!\! 1\!/2}f_{k}(S) \big).$$ In particular, for the vector $\xi$ whose components are all equal to one we have $f_{\xi=1}=f$.
Given a query from this family, the agency returns a perturbed response using a  Gaussian mechanism ${\cal M}_\xi$ by adding to $f_\xi$  a Gaussian vector  $U  \in\mathbb{R}^k$ where $U \sim N(0,\sigma^2 I)$, that is, $${\cal M}_\xi(S) = f_\xi(S) + U.$$ It is easy to see directly or from Lemma \ref{le:BW} that we can fix $\sigma^2$ and guarantee a given level of  DP$(\varepsilon,\delta)$ by choosing $\xi$ appropriately.
Hence  fixing $\sigma^2$ does not result in loss of generality. This is explained immediately following Theorem \ref{th: CIfix}. 
\subsection{Confidence regions}	\label{sec:cifix}

The following discussion concerns the choice of  $\xi \in \mathbb{R}^k_{>0}$ such that the standard confidence region
$CR^t_\xi$  for $\mu^*:=f(S)$ given in formula \eqref{eq:conf_regi_ellips_adjusted} below, which is based on the observed ${\cal M}_\xi(S)$, has the smallest volume. 
 It is easy to see that allowing the variance of $U$ in ${\cal M}_\xi$ to depend on $\xi$ does not lead to smaller volumes.
 
 The idea is simple: intuitively it appears efficient to add more noise to the more variable components of $f(S)$ rather than ``waste noise" on components with low variability. Note that  ``more variable" depends on both the population being measured and the chosen units of measurement. Instead of  asking the agency to adjust the noise to different components, we adjust the query, and thus the agency can use a standard Gaussian mechanism.
This intuition, as the whole paper, is clearly relevant only for $k>1$.

The analyst observes ${\cal M}_\xi(S) = f_\xi(S) + U$, where
\begin{equation*}
	{\cal M}_\xi(S) = \big(Diag(\xi)^{1/2}f(S) + U\big) \sim N(Diag(\xi)^{1/2}f(S), \sigma^2 I).
\end{equation*}
Thus, 
\begin{equation*}
Diag(\xi)^{-1/2}{\cal M}_\xi(S) =  \big( f(S) + Diag(\xi)^{-1/2}U\big)  \sim N(\mu^*, Diag(\xi)^{-1}\sigma^2),
\end{equation*}
where $\mu^*:=f(S)$.

The standard confidence region for $\mu_x$ based on $X \sim N(\mu_x, \Sigma)$ is $\{\mu:(X-\mu)^T\Sigma^{-1}(X-\mu)\le t\}$; see, e.g., \cite{anderson1962introduction}, p. 79.  Thus,
the confidence region for $\mu^*=f(S)$ based on $Diag(\xi)^{-1/2}{\cal M}_\xi(S)$  becomes
\begin{equation}\label{eq:conf_regi_ellips_adjusted}
	CR^t_\xi = \{ \mu \in \mathbb{R}^k: \big(Diag(\xi)^{-1/2}{\cal M}_\xi(S)-\mu \big)^T\\ (Diag(\xi)\sigma^{-2})  \big(Diag(\xi)^{-1/2}{\cal M}_\xi(S)-\mu \big) \leq t \}.
\end{equation}
For any $\xi \in \mathbb{R}_{>0}^k$ and any $\mu^* \in \mathbb{R}^k$,       the coverage probability $P(\mu^* \in CR^t_\xi)=P(Y\le t)$  where $Y \sim {\cal X}^2_k$ (the chi-square distribution with $k$ degrees of freedom), and thus all the regions $CR^t_\xi$ have the same confidence (coverage) level. 
We denote the volume by $Vol( CR^t_{\xi})$. For a discussion of the volume as a measure of utility of confidence regions see, e.g., \cite{efron2006minimum}.  We now need the notation 
 \begin{equation*}\label{eq:tilda}\psi \coloneqq f(\widetilde{S})-f(\widetilde{S}')=\big(f_1(\widetilde{S})-f_1(\widetilde{S}'),\ldots,f_k(\widetilde{S})-f_k(\widetilde{S}') \big), \end{equation*}
 where $(\widetilde{S},\widetilde{S}')$ is any pair in the set defined in \eqref{eq:tildeS},
 and we assume that $\psi_i^2>0$ for all $i$.

\begin{theorem} \leavevmode \label{th: CIfix} 
	
	{\rm({\bf 1})} For any fixed $t$, all confidence regions $CR^t_\xi$   defined in Equation  \eqref{eq:conf_regi_ellips_adjusted} have the same confidence level; that is, the probability $P(\mu^* \in CR^t_\xi)$  depends only on $t$ (and not on $\xi$).
	
	{\rm({\bf 2})} Set $$\Lambda(\xi)=\frac{\sqrt{\psi^T Diag(\xi) \psi}}{\sigma}.$$ If for two vectors $\xi_a$ and $\xi_b$ the mechanisms ${\cal M}_{\xi_a}$ and ${\cal M}_{\xi_b}$
	have the same level of {DP} (that is, the same $\epsilon$ and $\delta$) then
	$\Lambda(\xi_a)=\Lambda(\xi_b)$.
	
	{\rm({\bf 3})}	The choice $\xi = \xi^* := c\,(1/\psi_1^2,\ldots,1/\psi_k^2)$ with $c=||\psi||^2/k$ minimizes $ Vol( CR^t_\xi)$ for any $t>0$ over  all vectors $\xi \in \mathbb{R}_{>0}^k$ and associated mechanisms $\cal M_\xi$ having the same DP level. In particular,
	$$Vol( CR^t_{\xi^*}) \leq Vol( CR^t_{\xi=1}),$$
	with strict inequality when $max_i(\psi_i) \neq min_i(\psi_i)$. 
	The right-hand side of the inequality pertains to the query $f$.
\end{theorem}
Fix	$\sigma^2$ to be the smallest variance such that the Gaussian mechanism ${\cal M}_{\xi=1}(S)$ guarantees $DP(\varepsilon, \delta)$ for the query $f$.
Part (2) of Theorem \ref{th: CIfix} or of Theorem \ref{th:LRfix} below shows that for any ${\cal M}_\xi$ to have the same DP  level as ${\cal M}_{\xi=1}$ we must have
$\frac{\sqrt{\psi^T Diag(\xi) \psi}}{\sigma}=\frac{\sqrt{\psi^T \psi}}{\sigma}$.

For the proof we need two lemmas that are given first.

\begin{lemma}For any $\xi \in \mathbb{R}_{>0}^k$. 
	\label{le:ontildeS}
	\begin{equation*}
	\Delta(f_\xi) \equiv \operatorname*{max}_{S\sim S'\in {\mathfrak D}}||f_\xi(S)-f_\xi(S')||
	= ||f_\xi(\widetilde{S})-f_\xi(\widetilde{S}'))||,
\end{equation*}
where the pair $(\widetilde{S},\widetilde{S}')$ is defined in  Equation \eqref{eq:tildeS}.
\end{lemma}
The proof is given in the Appendix.

For an agency willing to release the query $f$, releasing $f_\xi$ under the mechanism $\cal M_\xi$ with the same DP level does not add any complications. The agency needs to compute the sensitivity defined by
$\Delta({f_\xi}) \equiv \operatorname*{max}_{S\sim S'\in {\mathfrak D}}||(f_\xi(S)-f_\xi(S')\big)||$. By Lemma \ref{le:ontildeS}, this amounts to computing
$||f_\xi(\widetilde{S})-f_\xi(\widetilde{S}'))||$ using the pair $(\widetilde{S},\widetilde{S}')$ from Equation \eqref{eq:tildeS}, which is needed to compute the sensitivity of $f$.
In particular,
$\Delta({f_\xi})= \sqrt{\psi^T Diag(\xi) \psi},$ and we shall see below that the quantity $\Delta({f_\xi})/{\sigma}$ is determined by the DP level. Given a DP$(\varepsilon,\delta)$ level, the agency guarantees it by choosing $\sigma$ using Lemma \ref{le:BW} below.

The proof of Theorem \ref{th: CIfix} (and all other theorems) relies on the next lemma; it can be obtained readily from the results of \cite{balle2018improving}, which hold for any query $f$. 
\begin{lemma}\label{le:BW}
	Let ${\cal M}(S)=f(S)+U$ be a Gaussian mechanism with $U\sim N(0, \sigma^2I)$,  and for  given datasets $S$ and $S'$ set $D := D_{S,S'} = ||f(S)-f(S')|| $. 
	{\rm({\bf 1})}  If
	\begin{equation}\label{ana_gaussi_requir}
		\Phi \left(\frac{D}{2\sigma}-\frac{\varepsilon\sigma}{D}\right) - e^{\varepsilon}\Phi \left(-\frac{D}{2\sigma}-\frac{\varepsilon\sigma}{D}\right) \leq \delta,
	\end{equation}
	then for all $E \subseteq \mathbb{R}^k$,
	\begin{equation}\label{condition_rdp}
		\mathbb{P}({\cal M}
		(S) \in E)\le e^\varepsilon \mathbb{P}({\cal M}(S') \in E)+\delta.
	\end{equation}
{\rm({\bf 2})}  Setting $\widetilde D := \Delta(f) = ||f(\widetilde S)-f(\widetilde S')||$, with $(\widetilde S,\widetilde S')$ given in Equation \eqref{eq:tildeS}, Equation \eqref{ana_gaussi_requir} holds with $D$ replaced by $\widetilde D$ if and only if the inequality \eqref{condition_rdp} holds for all $S \sim S'$ and $E \subseteq \mathbb{R}^k$, that is, if and only if DP$(\varepsilon,\delta)$ holds.
\end{lemma}
Part (2) of Lemma \ref{le:BW} coincides with Theorem 8 of \cite{balle2018improving}, and the first part follows from their method of proof.

\textit{Proof} of Theorem \ref{th: CIfix}.
Part (1) follows from the fact mentioned above that all these regions have confidence level
$P(Y\le t)$  where $Y\sim {\cal X}^2_k$.
Part (2)  is obtained by
replacing $f$ of Part (2) of Lemma \ref{le:BW} by $f_\xi$; then 
$(\widetilde D/\sigma)$ becomes $\frac{\sqrt{\psi^T Diag(\xi) \psi}}{\sigma}$ and the result follows.

To prove Part (3), note that the confidence region for the adjusted query given in Equation \eqref{eq:conf_regi_ellips_adjusted} is an ellipsoid whose volume is given by:
\begin{equation}\label{volume_adjusted_fixed}
Vol(CR^t_\xi) = V_k\cdot (\sigma^2 t)^{k/2}  \big({det[Diag(\xi)]}\big)^{-1/2}\,,
\end{equation}
where $V_k$ is the volume of the unit ball in $k$ dimensions. By Part (2), we have to minimize the volume as a function of $\xi$ subject to the constraint 
${\psi^T Diag(\xi)\psi}=  {\psi^T\psi}$, which we do by using Lagrange multipliers. See the Appendix for details.
 \qed
 
Given a DP level, the volume is minimized by choosing $\xi_i$ proportionally to $1/\psi^2_i$. Multiplying $\xi$ by a suitable constant guarantees the desired DP level.
It is easy compute the ratio of the volumes of the optimal region and the one based on the original query $f$:
\begin{equation*}\label{improv_diff_vol}
\frac{Vol(CR^t_{\xi^*})}{Vol(CR^t_{\xi=1})} 
= \left( \frac{\big(\prod_{i=1}^k \psi_i^2\big)^{1/k}}{\frac{1}{k}\sum_{i=1}^k \psi_i^2} \right)^{k/2}.
\end{equation*}
Clearly the ratio  is bounded by one, which can be seen again by the arithmetic-geometric mean inequality. 
Also, if one of the coordinates $\psi_i$ tends to zero, so does the ratio, implying the possibility of a substantial reduction in the volume obtained by using the optimally adjusted query $f_{\xi^*}$.  We remark that the ratio is decreasing in the partial order of majorization applied to $(\psi^2_1,\ldots,\psi^2_k)$; see \cite{MOA}.

\subsection{Testing hypotheses: Likelihood-ratio test}	
As in Section \ref{sec:cifix}, consider a query $f(S) \in \mathbb{R}^k$, which is observed with noise via a Gaussian privacy mechanism.  Now the analyst's  goal is to test the simple hypotheses
$ H_0: f(S)=0, \quad H_1: f(S)=\eta$. The null hypothesis is set at zero without loss of generality by a straightforward translation.  For any $\xi \in \mathbb{R}_{\ge 0}$ (a vector with nonnegative components), let $f_\xi(S)=Diag(\xi)^{1/2}f(S)$ and let
${\cal M}_\xi(S)=f_\xi(S)+U$,
where $U \sim N(0,\sigma^2 I)$ and  	$\sigma^2$ is the smallest variance such that the Gaussian mechanism ${\cal M}_{\xi=1}(S)$ guarantees $DP(\varepsilon, \delta)$ for the query $f$.

Let  $h_{\xi i}$ denote the density of ${\cal M}_\xi(S)$ under the hypothesis $H_i$, $i=0, 1$.  The log-likelihood ratio based on the observed ${\cal M}_\xi(S)$,
${\log}\big\{\frac{h_{\xi 1}({\cal M}_\xi(S))}{h_{\xi 0}({\cal M}_\xi(S))}\big\}$, is proportional to $\frac{{\cal M}_\xi(S)^T Diag(\xi)^{1/2} \eta}{\sigma^2}$, which under $H_0$ has the  $N(0, \frac{{\eta^{T}Diag(\xi)\eta}}{\sigma^2})$ distribution.
The likelihood-ratio test (which by the Neyman--Pearson lemma has a well-known optimality property)  rejects $H_0$ when the likelihood ratio is large. For a given significance level $\alpha$, the rejection region has the form
\begin{equation}
		\label{eq:rejreg}
		R_\xi=\Big\{ {\cal M}_\xi(S):  \frac{{\cal M}_\xi(S)^T Diag(\xi)^{1/2} \eta}{\sigma^2} > t\,\,\Big\}, \ \textnormal{where} \  t=\Phi^{-1}(1-\alpha) \frac{\sqrt{\eta^{T}Diag(\xi)\eta}}{\sigma}.
\end{equation}	

Let $\pi(R_\xi):=P_{H_1}({\cal M}_\xi(S) \in R_\xi)$ denote the power associated with the region $R_\xi$.

\begin{theorem}\label{th:LRfix}\leavevmode 
	
	{\rm({\bf 1})} For any fixed $\alpha$ and for all $\xi \in \mathbb{R}_{\ge 0}$, the rejection regions $R_\xi$ defined in \eqref{eq:rejreg} have significance level $\alpha$, that is, $P_{H_0}(R_\xi)=\alpha$.
	
	{\rm({\bf 2})} Assume that for two vectors $\xi_a$ and $\xi_b$ the mechanisms ${\cal M}_{\xi_a}$ and ${\cal M}_{\xi_b}$
	have the same level of DP (same $\epsilon$ and $\delta$); then
	$\varLambda(\xi_a)=\varLambda(\xi_b)$, where $\varLambda(\xi)=\frac{\sqrt{\psi^T Diag(\xi) \psi}}{\sigma}$.
	
	{\rm({\bf 3})}	Let $ j^* = \arg\max_i ( \eta_i^2/\psi_i^2),$
	and define $\xi^*$ by $\xi_{j^*}^*=||\psi||^2/\psi^2_{j^*}$  and $\xi^*_i = 0 \ \ \forall\, i\neq j^*$; then the choice $\xi=\xi^*$ maximizes the power $\pi(R_{\xi})$ 
	over  all vectors $\xi \in \mathbb{R}^k_{\ge 0}$ and the associated mechanisms $\cal M_\xi$ having the same DP level, and in particular 
	$\pi(R_{\xi^*}) \geq \pi(R_{\xi=1}), $
	with strict inequality unless $max_i( \eta_i^2/\psi_i^2) = min_i ( \eta_i^2/\psi_i^2 ).$
\end{theorem}
The right-hand side of the latter inequality pertains to the original query $f$, and thus a query of just one coordinate of $f$, the one having the largest ratio of (loosely speaking) signal ($\eta_i^2$) to noise ($\psi_i^2$) maximizes the power of the test.  
Note the difference between the  optimal query of Theorem \ref{th:LRfix} and that of Theorem \ref{th: CIfix},	which uses all coordinates of $f$.

\textit{Proof} of Theorem \ref{th:LRfix}. Part (1) follows from \eqref{eq:rejreg} and the discussion preceding it with standard calculations; Part (2) is similar to that of Theorem \ref{th: CIfix}. The proof of Part (3) is given in the Appendix.

\section{Random, normally distributed data}	\label{sec:randomdata}
So far the dataset $S$ was considered fixed, that is, nonrandom.  Statisticians often view the data as random,  model the data-generating process, and study the model's parameters. Accordingly, we now assume that the dataset, denoted by $T$, is randomly  generated as follows: the rows of $T$, $T_1,\ldots,,T_n$ are iid, where each row $T_\ell \in \mathbb{R}^d$  represents $d$ measurements of an individual in the random sample $T$. We also assume that $f$ is a linear query, that is,
$$
	f(T) = \left( f_1(T),...,f_k(T) \right)
	= \Big( \frac{1}{n}\sum_{\ell=1}^n q_1(T_\ell),...,\frac{1}{n}\sum_{\ell=1}^n q_k(T_\ell)\Big)
$$
for some functions $q_1,\ldots,q_k$. In addition, we assume that 
$q(T_\ell) := \big(q_1(T_\ell),\ldots,q_k(T_\ell)\big) \sim N(\mu^*, \Sigma)$ 
for some unknown $\mu^*$ and a known covariance matrix $\Sigma$. The normality assumption holds when the entries of $T$	are themselves normal, and $q_i$ are  linear functions. Assuming normality, possibly after transformation of the data, and iid observations is quite common in statistical analysis. It follows that $f(T)\sim N(\mu^*, \Sigma_n)$, where $\Sigma_n=\Sigma/n$. This may hold approximately by the central limit theorem even if normality of the dataset is not assumed. Here we assume that $\Sigma$ is known. The case where it is obtained via a privatized query is beyond the scope of this paper. 
Assuming that $\Sigma$ is known is sometimes natural. For example, when we
test  hypotheses on means of subpopulations, we sometimes use the covariance matrix estimated
from the general population. In annual economic surveys, for example, the focus is on change between consecutive years in, say, average income or unemployment rate;  however, one can assume that the past years' covariance matrix is roughly unchanged.  In
Section \ref{sec:numst} we give an example of blood-test data, where we use a correlations matrix estimated from the general population.

Since the observed data will depend  only on $q(T_\ell)$, we  now redefine the dataset to be $S$, consisting of the $n$ iid rows $S_\ell :=q(T_\ell)$, \, $\ell=1, \ldots,n$.
The assumption $q(T_\ell)\sim N(\mu^*, \Sigma)$ implies that these rows can take any value in $C:=\mathbb{R}^k$.  The universe of all such matrices $S$ is   $\mathfrak D := C^n =\mathbb{R}^{n \times k}$.

Our goal is to construct a confidence region for the model parameter  $\mu^*$ and test hypotheses about it. This can be done via the query $f(S)=\frac{1}{n}\sum_{\ell=1}^nS_\ell$ having the distribution $N(\mu^*, \Sigma_n)$; however, we show that posing the query  $g(S):=\Sigma_n^{-1/2}f(S)$ under the same Random Differential Privacy parameters (RDP, to be defined below) results in smaller confidence regions. We also compare the powers of certain tests of hypotheses.

We say that a  query $f$ is invariant  if $f(S)$ is invariant under permutations of the rows of $S$. This happens trivially when $f$ is a linear query as defined above.
If $f$ is invariant then the distribution of the output of any mechanism that operates on $f$ is obviously unchanged by permutations of rows. In this case it suffices to consider neighbors      $S \sim S'$ of the form
$ S=(S_1,\ldots,S_{n-1},S_n), \ \ \ S'=(S_1,\ldots,S_{n-1},S_{n+1}).$
We assume that $S_1,\ldots,S_{n+1}$ are iid rows having some distribution $Q$.
\begin{definition} \label{def:RDP}
	$(\varepsilon, \delta, \gamma$)-{\rm Random Differential Privacy (\cite{hall2013random})}. 
	A random perturbation mechanism $\cal M$ whose distribution is invariant under permutations of rows is said to be ($\varepsilon, \delta, \gamma$)-{Randomly Differentially Private}, denoted by RDP$(\varepsilon, \delta, \gamma)$, if
	\begin{equation*} {P}_{S_1,\ldots,S_{n+1}} \Big(\forall \ E \subseteq \mathbb{R}^k, \, \,  {P}({\cal M}(S) \in E|S) \\ \le e^\varepsilon \mathbb{P}({\cal M}(S') \in E|S')+\delta \Big)  \geq 1-\gamma,\end{equation*}
	where
	$S$ and  $S'$ are neighbors as above, the   probability ${P}_{S_1,\ldots,S_{n+1}}$ is with respect to  $S_1,\ldots,S_{n+1}\overset{iid}{\sim} Q$ and the probability ${P}({\cal M}(S) \in E|S)$ refers to the noise $U$ after conditioning on $S$.
	\end{definition} 
In words, instead of requiring the condition of differential privacy to hold for all $S\sim S' \in {\mathfrak D}$, we require that there be a  ``privacy set" in which any  two random neighboring datasets satisfy the DP condition, and its probability is bounded below by $1-\gamma$. An objection to this notion may arise from  the fact that under RDP ``extreme" participants, who are indeed rare, are not protected, even though they may be the ones who need privacy protection the most.
Since RDP is not in the worst-case analysis spirit of DP, we remark that DP can be obtained if, instead of  
ignoring worst cases having small probability as in RDP, the agency trims them by either removing them from the dataset or by projecting them to a given ball (that is independent of the dataset) which determines the sensitivity. Such trimming, if its probability is indeed small, corresponding to a small $\gamma$, will not overly harm the data analysis.

 To define a mechanism ${\cal M}_h(S)=h(S) +U$ (see \eqref{eq:DPDP}) that satisfies RDP$(\varepsilon, \delta, \gamma)$, we need to define a ``privacy set" $H$, which is a subset of $\mathfrak D \times \mathfrak D$ consisting of neighboring pairs $(S,S')$, that satisfies two conditions. \textbf{(A)}:  $P((S,S')\in H)=1-\gamma$, where  the probability $P$ is ${P}_{S_1,\ldots,S_{n+1}}$ of Definition \ref{def:RDP}, and \textbf{(B)}:  Equation \eqref{eq:DPDP} holds for all $E$, and  any pair of neighboring datasets $(S,S') \in H$.
 We then say that ${\cal M}_h^H(S)$ is RDP$(\varepsilon, \delta, \gamma)$ with respect to the privacy set $H$ and the query $h$.

 To construct a suitable $H$ satisfying condition \textbf{(A)} note that 
\begin{equation}\label{eq:disff}
f(S)-f(S')= \frac{1}{n}[q(S_n)-q(S_{n+1})] \sim N(0, 2\Sigma/n^2)
\end{equation} 
and
$$ ||g(S)-g(S')\,||^2=||\Sigma_n^{-1/2}[f(S)-f(S')]||^2  =||\Sigma^{-1/2}[q(S_n)-q(S_{n+1})]\,||^2/n \sim 2{\cal X}_k^2/n.
$$
Thus, if $Y \sim {\cal X}^2_k$ satisfies $P(Y \le r^2)=1-\gamma$ then 
$P\big(||g(S)-g(S')\,||^2 \le  2r^2/n \big)= 1-\gamma$, and we can choose the set $H$ to be $$H_g:=\big\{(S,S') \in \mathfrak{D}\times \mathfrak{D} \,:\,||g(S)-g(S')\,||^2 \le  2r^2/n\big\}.$$
Also, by \eqref{eq:disff}, we have that $||f(S)-f(S')\,||^2$ is distributed as $(2/n)Z^T\Sigma_n Z$ with $Z \sim N(0,I)$, and it is well known (by diagonalizing $\Sigma_n$) that the latter expression has the  
$(2/n)\sum_{i=1}^k \lambda_i X_i$ distribution, where $\lambda_i$ denote the eigenvalues of $\Sigma_n$ and $X_i$ are iid ${\cal X}^2_1$.

Another privacy set  we consider is given by 
$$H_f:=\big\{(S,S')\in \mathfrak{D}\times \mathfrak{D}\,:\,||f(S)-f(S')\,||^2 \le  C^2\big\},$$
where $C$ is such that $P((S,S')\in H_f)=1-\gamma$, and by \eqref{eq:disff} the constant $C$ depends on $\Sigma$ and $n$.

We consider three Gaussian mechanisms: 
\begin{align*}
{\cal M}^{H_g}_g(S)&=g(S)+U, \text{  where  } U\sim N(0, \sigma_g^2 I),  \\
{\cal M}^{H_g}_f(S)&=f(S)+U \text{  with  } U \sim N(0,\sigma_{fg}^2 I),\\
{\cal M}^{H_f}_f(S)&=f(S)+U \text{  with  } U \sim N(0,\sigma_f^2 I),
\end{align*}
where the first two are with respect to the privacy set  $H_g$, and the third is with respect to $H_f$.
For each of the three, an appropriate noise variance $\sigma^2_g$, $\sigma^2_{fg}$, and  $\sigma_f^2$ has to be computed, given the privacy set and the RDP parameters, so that condition \textbf{(B)} above holds.
To determine the noise variance we have to compute the sensitivity of the query $g$ on the set $H_g$ and the sensitivity of $f$ on both $H_g$ and $H_f$.

Define the sensitivity of $f$ and $g$ on $H_g$, denoted by $D(fg)$ and ${D}(g)$, respectively, and the sensitivity of $f$  on $H_f$, denoted by $D(f)$, as follows:
\begin{align}\label{eq:DfDg}
	&D(fg) := \operatorname*{max}_{\ (S, S') \in H_g}   ||f(S)-f(S')||, \nonumber\\  
	&{D}(g) := \operatorname*{max}_{\ {(S, S')} \in H_g}   ||g(S)-g(S')|| =\sqrt{2}\, r/\sqrt{n}\,,\\
	&D(f):= \operatorname*{max}_{\ {(S, S')} \in H_f}   ||f(S)-f(S')|| =C. \nonumber \end{align}

We compare the above three mechanisms with the same RDP level
in terms of the volume of confidence regions  and the power of tests of hypotheses  for the model parameter  $\mu^*$, computed from data given by  these mechanisms.
We 	shall prove in Sections \ref{sec:CIrand} and \ref{sec:testrand} that the mechanism ${\cal M}^{H_g}_g(S)$ is better than ${\cal M}^{H_g}_f(S)$  in terms of the volumes of confidence regions and the power of tests. It is easy to see that $D(f) \le D(fg)$ and we show below that this implies that ${\cal M}^{H_f}_f(S)$ is better than ${\cal M}^{H_g}_f(S)$ both in terms of the volume of confidence regions, and the power of tests of simple hypotheses. We shall also show that for small $\gamma$ we have $D(g) \le D(f)$, and that ${\cal M}^{H_g}_g(S)$ is better than ${\cal M}^{H_f}_f(S)$ in terms of volume of confidence regions.
The latter mechanism is discussed in Section \ref{sec:fHf}.	

 By the definition of RDP, the mechanism ${\cal M}^{H_g}_f(S)$ satisfies $RDP(\varepsilon, \delta, \gamma)$ when \eqref{ana_gaussi_requir} holds with $D=D(fg)$ and  $\sigma=\sigma_{fg}$, as does the mechanisms ${\cal M}^{H_g}_g(S)$ with $D=D(g)$ and  $\sigma=\sigma_g$, and likewise the mechanism ${\cal M}^{H_f}_f(S)$ with $D=D(f)$ and $\sigma=\sigma_f$.
\begin{lemma}\label{le:RDPlemma}
	If  ${\cal M}^{H_g}_g$, ${\cal M}^{H_g}_f$, and ${\cal M}^{H_f}_f$ have the same RDP, then $D(g)/\sigma_g = D(fg)/\sigma_{fg} = D(f)/\sigma_f$. The first equality is equivalent to  $\sigma_{fg}^2=\lambda_{max}(\Sigma_n)\sigma^2_g$, where $\lambda_{max}(\Sigma_n)$ denotes the largest eigenvalue of $\Sigma_n$.
\end{lemma}
\begin{proof} The first part follows from Lemma \ref{le:BW} and the above discussion.
For the second part it suffices to prove that $[D(fg)]^2=\lambda_{max}(\Sigma_n)[D(g)]^2$. To see this note that the maximization in \eqref{eq:DfDg} is equivalent to maximizing $(g(S)-g(S'))^T \Sigma (g(S)-g(S'))$ subject to $||g(S)-g(S')||^2=[D(g)]^2$, and the result follows readily from Rayleigh's theorem; see, e.g., \cite{horn2012matrix}, Chapter 4.
\end{proof}

\subsection{Confidence regions}\label{sec:CIrand}
We have 
\begin{equation*}\Sigma_n^{1/2}{\cal M}^{H_g}_g(S) \sim N \left(\mu^*, \Sigma_n(1+\sigma_g^2) \right), \ \ {\cal M}^{H_g}_f(S) \sim N(\mu^*, \Sigma_n+\sigma^2_{fg} I), \ \ {\cal M}^{H_f}_f(S) \sim N(\mu^*, \Sigma_n+\sigma_f^2 I).\end{equation*}
The standard confidence regions for $\mu^*:=E[f(S)]$ based on ${\cal M}^{H_g}_g(S)$, ${\cal M}^{H_g}_f(S)$, and ${\cal M}^{H_f}_f(S)$ are 
\begin{align*}\label{eq:conf_regi_ellips_per}
CR^t_{g} &= \Big\{ \mu \in \mathbb{R}^k: \big(\Sigma_n^{1/2}{\cal M}^{H_g}_g(S)-\mu \big)^T (\Sigma_n(1+\sigma_g^2))^{-1} \big(\Sigma_n^{1/2}{\cal M}^{H_g}_g(S)-\mu \big) \leq t \Big\},\\
	CR^t_{fg} &= \Big\{ \mu \in \mathbb{R}^k: \big({\cal M}^{H_g}_f(S)-\mu \big)^T (\Sigma_n+\sigma_{fg}^2 I)^{-1} \big({\cal M}^{H_g}_f(S)-\mu \big) \leq t \Big\},\\
			CR^t_{f} &= \Big\{ \mu \in \mathbb{R}^k: \big({\cal M}^{H_f}_f(S)-\mu \big)^T (\Sigma_n+\sigma_f^2 I)^{-1} \big({\cal M}^{H_f}_f(S)-\mu \big) \leq t \Big\}.
	\end{align*}
	
	The next theorem shows that confidence regions based on ${\cal M}^{H_g}_g$ have a smaller volume than those based on ${\cal M}^{H_g}_f$,  and, for $\gamma$ sufficiently small, also than those based on ${\cal M}^{H_f}_f$. 
	Thus, of the three natural candidates we consider, ${\cal M}^{H_g}_g$ is the best mechanism for small $\gamma$. 

\begin{theorem} \leavevmode\label{th: CIrnd} 
	
	{\rm({\bf 1})} For any fixed $t$, the confidence regions $CR^t_{g}$,  $CR^t_{fg}$, and $CR^t_{f}$ have the same confidence level; that is, for any $\mu^*$ we have $P(\mu^* \in CR^t_{g})=P(\mu^* \in CR^t_{fg})=P(\mu^* \in CR^t_{f})$.
	
	{\rm({\bf 2})} If the mechanisms ${\cal M}^{H_g}_g$, ${\cal M}^{H_g}_f$, and ${\cal M}^{H_f}_f$ have the same level of $RDP(\varepsilon, \delta, \gamma)$ then $D(g)/\sigma_g = D(fg)/\sigma_{fg}=D(f)/\sigma_f$.
	
	{\rm({\bf 3})} $Vol( CR^t_{g}) \leq Vol( CR^t_{fg}),$
	with strict inequality unless all the eigenvalues of $\Sigma_n$ are equal.
	
	{\rm({\bf 4})} For sufficiently small $\gamma$, $Vol( CR^t_{g}) \leq Vol( CR^t_{f}),$
	with strict inequality, unless all the eigenvalues of $\Sigma_n$ are equal.
\end{theorem}
\textit{Proof}. \,\,
	Part (1) holds as in Theorem \ref{th: CIfix}, and Part (2) holds by Lemma \ref{le:RDPlemma}. The proof of Part (3), given in the Appendix, uses the relation $\sigma_{fg}^2=\lambda_{max}(\Sigma_n)\sigma^2_g$ of Lemma \ref{le:RDPlemma}, and a straightforward eigenvalue comparison. The proof of Part (4) is somewhat more involved. It uses a comparison of distribution functions of weighted sums of independent gamma random variables and a majorization argument. Details and references are given in the Appendix.

\subsection{Testing hypotheses: Likelihood-ratio test}\label{sec:testrand}

With $E[f(S)]=\mu^*$ we consider the hypotheses $ H_0: \mu^*=0$ and $H_1: \mu^*=\eta$ and the mechanisms ${\cal M}^{H_g}_f(S)$ and ${\cal M}^{H_g}_g(S)$ defined above. If ${\cal M}^{H_g}_f(S)$ is observed then the rejection region $R_{fg}$ of the likelihood-ratio test with significance level $\alpha$ has the form
\begin{equation*}\label{eq:optim_rej_regio_gaus_data}
	R_{fg} =\big\{ {\cal M}^{H_g}_f(S):  {\cal M}^{H_g}_f(S)^T(\Sigma_n+\sigma_{fg}^2 I)^{-1}\eta > t \,\,\big\}, \ \textnormal{where} \  t=\Phi^{-1}(1-\alpha) \sqrt{\eta^T(\Sigma_n+\sigma_{fg}^2 I)^{-1}\eta}.
\end{equation*}
If ${\cal M}^{H_g}_g(S)$ is observed then the testing problem becomes $ H_0: \mu^*=0$   vs. $ H_1: \mu^*=\Sigma_n^{-1/2}\eta$, and the rejection region $R_g$ of the likelihood-ratio test with significance level $\alpha$ has the form
\begin{equation*}
	R_g =\big\{ {\cal M}^{H_g}_g(S):  {\cal M}^{H_g}_g(S)^T[(1+\sigma_g^2)I]^{-1}\Sigma_n^{-1/2}\eta > t \big\}, \ \textnormal{where} \  t=\Phi^{-1}(1-\alpha) \sqrt{\frac{\eta^T\Sigma_n^{-1}\eta}{\sigma_g^2+1}}.
\end{equation*}

\begin{theorem} \leavevmode \label{th:LRRnd}
	
	{\rm({\bf 1})} The rejection regions $R_{fg}$ and $R_g$  have the same significance level $\alpha$.
	
	{\rm({\bf 2})} If both mechanisms ${\cal M}^{H_g}_g$ and ${\cal M}^{H_g}_f$ have the same level of $RDP(\varepsilon, \delta, \gamma)$ then $D(g)/\sigma_g = D(fg)/\sigma_{fg}$.
	
	{\rm({\bf 3})}	Let $\pi(R_g)$ and $\pi(R_{fg})$ denote the power associated with the rejection regions $R_g$ and $R_{fg}$, respectively; then
	$ \pi(R_g) \geq \pi(R_{fg})$
	with strict inequality, unless all the eigenvalues of $\ \Sigma_n$ are equal.
\end{theorem}
\begin{proof}
Part (1) is similar to Part (1) of Theorem \ref{th:LRfix}.	Part (2) is already given in Theorem \ref{th: CIrnd}.
       The proof of Part  (3), given in the Appendix, involves a simultaneous diagonalization argument and a comparison of eigenvalues using $\sigma_{fg}^2=\lambda_{max}(\Sigma_n)\sigma^2_g$.
       \end{proof}
  
   \subsection{The mechanism ${\cal M}^{H_f}_f(S)$}\label{sec:fHf}
   It is easy to see that the sensitivity $D(f)$ (on $H_f$) satisfies $D(f) \le D(fg)$. An equal RDP level implies by Lemma \ref{le:RDPlemma} that $D(fg)/\sigma_{fg}=D(f)/\sigma_f$ and therefore $\sigma_{fg} \ge \sigma_f$.
   The power of the likelihood-ratio test based on 
   ${\cal M}^{H_f}_f(S)$ is the same as in \eqref{eq:opt_pow_gaus_data}, with $\sigma_{fg}$ replaced now with $\sigma_f$. This power is easily seen to be decreasing in $\sigma$ and therefore the likelihood-ratio test based on 
   ${\cal M}^{H_f}_f(S)$ has a higher power than the test  based on ${\cal M}^{H_g}_f(S)$. On the other hand, an extensive computational study shows that the power of the test  based on ${\cal M}^{H_g}_g(S)$ may be higher or lower than that based on ${\cal M}^{H_f}_f(S)$ depending on the parameters involved. 
   In the case of Figure 1 we see that the test  based on ${\cal M}^{H_g}_g(S)$ has a higher power than that based on ${\cal M}^{H_f}_f(S)$.

\subsection{Four degrees of naivet\'{e}}

We consider the random query  $f(S)\sim N(\mu^*, \Sigma_n)$. The discussion below applies in principle to other distributions.  We discuss four ways of testing hypotheses on $\mu^*$ in the presence of perturbation noise having a known distribution. The discussion pertains either to a DP$(\varepsilon,\delta)$ mechanism ${\cal M}_f$, see \eqref{eq:DPDP}, or to an RDP$(\varepsilon,\delta,\gamma)$ mechanism ${\cal M}^H_f$ with some privacy set $H$, which together with the RDP parameters determines the variance of the added noise.

\begin{enumerate}
	
	\item The most naive approach to analyzing perturbed data is to ignore the added noise altogether and determine  rejection regions and significance levels as if $f(S)$ is observed without noise. In this case the test may not be optimal, and the significance level will be wrong. We call this approach \textit{super-naive analysis}.
	
	\item A less naive approach is to choose a test based as above on the wrong assumption that $f(S)$ is observed without noise, but to set its critical value $t$, which determines the  significance level, according to the correct distribution of the observed data, taking the noise into account. In Figure 1 we depict  the power curve (denoted by $[\textbf{a}]$ in the example below) of this approach where RDP is with respect to the privacy set $H_f$. We call this approach \textit{naive analysis}.
	
 	\item An even better approach is to choose the test  optimally by computing the likelihood-ratio test and the significance level using the correct distribution of the observed  response, taking the Gaussian noise into account. For simple queries and Gaussian noise this approach is feasible analytically.   In Figure 1 we show two curves of the power under this approach, associated with ${\cal M}^H_f$ when RDP is with respect to the privacy set $H_g$ (denoted by $[\textbf{b}]$) and $H_f$ (denoted by  $[\textbf{c}]$), respectively.
 	We call this approach \textit{optimal analysis}.
	
	\item In this paper we propose  adjusting the query to the statistical goals of the analyst and then using the optimal rejection or confidence regions based on the observed response to the adjusted query, and on its correct distribution, taking the adjustment and the noise into account. This is accomplished by the mechanism  ${\cal M}^{H_g}_g(S)$ discussed above. Its power curve  is denoted by $[\textbf{d}]$ in Figure 1. We call this approach  \textit{adjusted optimal analysis}.
	
\end{enumerate}

The first two approaches are sometimes used by practitioners; they may be acceptable for very large samples. Their properties have been studied asymptotically.

\section{A numerical example}\label{sec:numst}

We provide a simple data example. The privacy of medical data is of utmost importance. Consider a dataset consisting of blood-test results. A standard blood test contains 30-40 variables measured in different units, with ranges and variances that are very different. Some of these variables are highly correlated. Of the many blood-test measurements, we chose for the sake of our examples to focus on six variables: 
Cholesterol, High Density Lipoprotein, Apo Protein A-1, Low Density Lipoprotein, Total Lipid, and Glucose  (all having the same units, MG/DL).

In this order, their covariance matrix $\Sigma$ is given below. It is based on data from \cite{qureshi2017application, castelli1977distribution, blum1985relationship}, and other Internet medical sources. Clearly, our goal is to provide a simple example to make our point, and not as a study of how to protect blood-test data. 
$$ {\Sigma} =
\begin{pmatrix*}[r] 
	1600\, & -160\, & -400\, & 840\, & 800\, & -40\\
	*\, & 400\, & 160\, & -175\, & -200\, & 0\\
	* & * & 1600\, & 280\,  & 600\, & 0\\
	* & * & * & 1225\, & 700\, & -35\\
	* & * & * &* & 2500 & -50\\
	* & * & * & * & * & 100 \\
\end{pmatrix*}
\quad
$$

We consider the release of averages of the above six variables over a sample of size $n$, which will vary in our examples. It is quite standard for statisticians to assume (with justification by the central limit theorem) that such vectors of averages are multivariate normal. We assume the agency releases data under $RDP(\varepsilon, \delta, \gamma)$. If instead of RDP we use DP and trimming as described in Section \ref{sec:randomdata}, we obtain essentially the same results, with DP$(4\varepsilon, \delta)$.

In the examples below we consider various parameters and alternatives.
\begin{table}[htbp]
	\caption{Examples 1--4} \label{tab:freq}
	\begin{center}
	\resizebox{0.55\textwidth}{!}{	
	\begin{tabular}{ccccc}
		\textbf{Example} &\textbf{ $\eta$} & \textbf{$n$} & \textbf{$\delta$} & \textbf{$\gamma$}\\
		\hline
		(1)&(10, 5, 10, 8.75, 12.5, 2.5)&50&$0.0200$&$10^{-4}$\\
		(2)&(10, 5, 10, 8.75, 12.5, 2.5)&50&$0.0004$&$10^{-6}$\\
		(3)&(0, 0, 20, 0, 25, 5)&50&$0.0004$&$10^{-6}$\\
		(4)&(0, 0, 20, 0, 25, 5)&100&$0.0001$&$10^{-6}$\\
	\end{tabular}}
\end{center}
\end{table}

\vspace{-0.5cm}

\begin{figure}[H]
  \centering
  \includegraphics[width=0.55\linewidth]{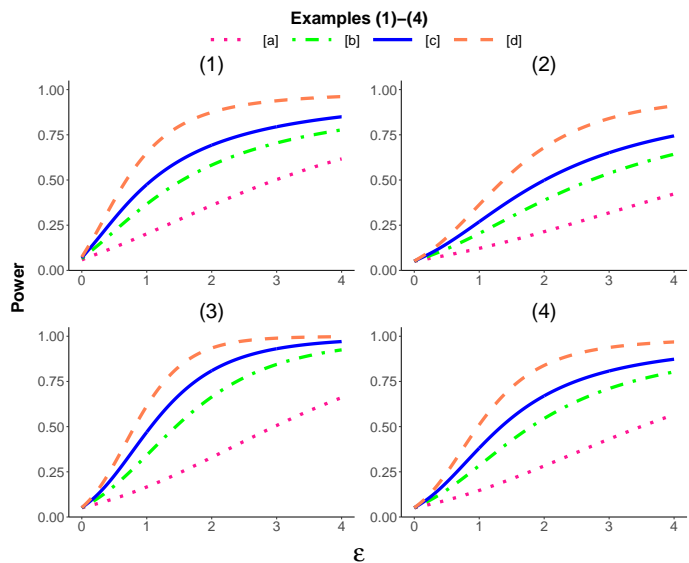}
  \caption{Comparison of power of the likelihood-ratio test with significance level $\alpha=0.05$ as a function of\, $\varepsilon$\, for naive, optimal, and adjusted optimal analyses.}
  \label{fig:fig1}
\end{figure}

\section*{Acknowledgments}
This work was supported in part by a gift to the McCourt School of Public Policy and Georgetown University, Simons Foundation Collaboration
733792, and Israel Science Foundation (ISF) grants 1044/16 and 2861/20.

We are grateful to Katrina Ligett and Moshe Shenfeld for very useful
discussions and suggestions.

\bibliographystyle{plain}  
\bibliography{Main}  

\newpage

\section{Appendix}

\textbf{\textit{Proof} of Theorem \ref{th: CIfix}} Part (3). The confidence region for the adjusted query given in \eqref{eq:conf_regi_ellips_adjusted} is an ellipsoid whose volume is given by
\begin{equation*}
Vol(CR^t_\xi) = V_k \cdot(\sigma^2 t)^{k/2}  \big({det[Diag(\xi)]}\big)^{-1/2}\,,
\end{equation*}
where $V_k$ is the volume of the unit ball in $k$ dimensions. In view of Part (2) we minimize the log of the volume as a function of $\xi$ subject to the constraint 
${\psi^T Diag(\xi)\psi}=  {\psi^T\psi}$.
We consider the Lagrangian
\begin{equation*}
{\mathcal L}(\xi_1,\ldots\xi_k, \lambda)
= - \sum_{i=1}^k \log \left(\xi_i \right) - \lambda \big[\sum_{i=1}^k\psi_i^2\xi_i - \sum_{i=1}^k\psi_i^2 \big]. 
\end{equation*}
We are minimizing a strictly convex function subject to a linear constraint. Differentiating and setting the Lagrangian to zero we readily obtain the unique minimum when $\xi_i$ is proportional to $1/\psi_i^2$. The constraint 
${\psi^T Diag(\xi) \psi}=  {\psi^T\psi}$, which by Part (2) guarantees the same DP level, implies that $\xi^*=c\,(1/\psi^2_1,\ldots,1/\psi^2_k)$ with $c=||\psi||^2/k$.
\qed\\

\textbf{\textit{Proof} of Lemma \ref{le:ontildeS}}.\quad
	It is easy to see that 	
	$\Delta^2(f_\xi) =
	\operatorname*{max}_{S\sim S'\in {\mathfrak D}}\sum_{i=1}^k \xi_i(f_i(S)-f_i(S'))^2.$
	We assumed that the components of the vector query $f= \left(f_1,\ldots,f_k\right)$ are functions of disjoint sets of columns of $S$. Since ${\mathfrak D}={\mathcal C}_1\times\ldots\times {\mathcal C}_d$,  the sum is maximized by  maximizing each summand individually. Multiplying each summand by a positive constant does not change the point where the maximum is achieved. \qed\\

\textbf{\textit{Proof} of Theorem \ref{th:LRfix} Part (3)}. Note first that for ${\cal M}_\xi$ and ${\cal M}_{\xi=1}$ to have the same DP$(\varepsilon,\delta)$ level we must have
$$\varLambda(\xi)=\frac{\sqrt{\psi^T Diag(\xi) \psi}}{\sigma}=\varLambda(1)=\frac{\sqrt{\psi^T \psi}}{\sigma}.$$
The power of the rejection region $R_\xi$ is 
\begin{equation*} \label{eq:power_adjusted_like_ratio}
\pi(R_\xi) = \, P_{H_1} \left( \frac{{\cal M}_\xi(S)^T \xi^{1/2} \eta}{\sigma^2} >            \Phi^{-1}(1-\alpha) \frac{\sqrt{ \eta^{T}Diag(\xi)\eta}}{\sigma} \right)
= 1-  \Phi \left( \Phi^{-1}(1-\alpha) - \frac{\sqrt{ \eta^{T}Diag(\xi)\eta}}{\sigma} \right),
\end{equation*}
which is increasing in ${ \eta^{T}Diag(\xi)\eta}$. Thus in order to maximize the power we have to maximize ${ \eta^{T}Diag(\xi)\eta}$        over $\xi$ subject to $\psi^T Diag(\xi) \psi={\psi^T \psi}$. Defining $v_i=\xi_i \psi^2_i$ the problem now is to maximize $\sum_i v_i\frac{\eta_i^2}{\psi_i^2}$ over $v_i \ge 0$, subject to $\sum_i v_i = ||\psi||^2$. Clearly the maximum is attained when $v_{j^*}=||\psi||^2$, where $j^*  = \arg\max_i ( \eta_i^2/\psi_i^2),$ and $v_i=0$ for $i \neq j^*$, 
completing the proof. \qed\\

\textbf{\textit{Proof} of Theorem \ref{th: CIrnd}}.
The proof of Part (3) uses the fact that
$$Vol( CR^t_{fg}) = b_k{\sqrt{det[\Sigma_n+\sigma_{fg}^2 I]}} \ \ \textnormal{and} \ \ 
Vol( CR^t_{g}) = b_k \sqrt{det[\Sigma_n(1+\sigma_g^2)]},$$
where $b_k=t^{k/2} V_k$ and \,$V_k$ is the volume of the $k$-dimensional  unit ball, and the relation $\sigma_{fg}^2=\lambda_{max}(\Sigma_n)\sigma^2_g$. The required inequality follows from
the relations
$$det \left(  \Sigma_n(\sigma_g^2+1) \right)
= \prod_{i=1}^{k} \left[ \frac{\lambda_i}{\lambda_{max}(\Sigma_n)}\sigma_{fg}^2+\lambda_i \right] 
\leq \prod_{i=1}^{k} \left[ \sigma_{fg}^2+\lambda_i \right] = det \left(  \Sigma_n+\sigma_{fg}^2 I\right),$$ where $\lambda_i$ denote the eigenvalues of $\Sigma_n$. 

To prove Part (4), we need the following fact, which is a special case of a result stated in \cite{DiacPerl}, Proposition 2.7 and Equation (10). The last part is given by \cite{Roosta-Khorasania}, p. 999, and \cite{Roosta-Khorasani}, Theorem 2.2.

\vspace{0.3cm}

\textbf{Fact}. \textit{Let $X_i\sim {\cal X}^2_1$ be iid. Without loss of generality assume that  $\lambda_1,\ldots,\lambda_k$ with $\lambda_i>0$ satisfy $\overline \lambda:=\sum_{i=1}^k \lambda_i/k=1$. Define $F_\lambda(x)=P(\sum_{i=1}^k \lambda_i X_i \le x)$ and let $F(x)$ denote the distribution function of ${\cal X}^2_k$.  Then for sufficiently large $x$ we have $F_\lambda(x)\le F(x)$.} 
\vspace{0.3cm}

More specifically, the latter inequality holds for $x>2k$.The latter lower bound, given by \cite{Roosta-Khorasania, Roosta-Khorasani}, is far from being tight, as suggested by numerical computations.

Recall that $C=D(f)$ satisfies $P\big((2/n)\sum_{i=1}^k \lambda_i X_i \le C^2\big)=1-\gamma$. For the rest of the proof set $D=D(g)$; then $D$ is defined by 
$P\big((2/n)Y \le D^2\big)=1-\gamma$, where $Y \sim {\cal X}^2_k$; see Equation \eqref{eq:DfDg}. For $\overline \lambda=1$, which can be assumed without loss of generality, the above Fact immediately implies that $C^2 \ge D^2$ for sufficiently small $\gamma$. By the last part of the above fact, for $k=6, 10$, 20, and 30, sufficiently small means $\gamma  \le 1-P(Y <12 )=0.062$ and $\gamma  \le 0.03$, 0.005, and 0.001, respectively. 
\vspace{0.3cm}

\textbf{\textit{Proof} of Part (4)}. As in the proof of Part (3), by Lemma \ref{le:RDPlemma} and then for sufficiently small $\gamma$ such that $C^2 \ge D^2 \overline \lambda$ (where $\overline \lambda=1$),  we have
$det \left(\Sigma_n(\sigma_g^2+1) \right)
= \prod_{i=1}^{k} \left[\lambda_i+ \frac{\lambda_i D^2}{C^2}\sigma_f^2 \right] 
\leq \prod_{i=1}^{k} \left[\lambda_i + \lambda_i\sigma_f^2/\overline \lambda\right]$ and now it remains to show that the latter product is bounded above by $\prod_{i=1}^{k} [\lambda_i+\sigma_f^2]=det \left(  \Sigma_n+\sigma_f^2 I\right)$.
Dividing by $\prod_{i=1}^{k}\lambda_i$ and taking log, we see that the required bound is equivalent to $\sum_{i=1}^{k}\log\left[1+ \sigma_f^2/\overline \lambda \right] \le \sum_{i=1}^{k} \log[1+\sigma_f^2/\lambda_i]$. This follows from the fact that $\log[1+\sigma_f^2/\lambda]$ is convex in $\lambda$ and therefore  $\sum_{i=1}^{k}\log[1+\sigma_f^2/\lambda_i]$ is a Schur convex function; see \cite{MOA}.
\qed\\

\textbf{\textit{Proof} of Theorem \ref{th:LRRnd} Part (3)}.
The power of the  rejection region $R_{fg}$ is given by
\begin{align}\label{eq:opt_pow_gaus_data}      \pi(R_{fg})& = P_{H_1} \Big({\cal M}^{H_g}_f(S)^T(\Sigma_n+I\sigma_{fg}^2)^{-1}\eta \nonumber > \Phi^{-1}(1-\alpha) \sqrt{\eta^T(\Sigma_n+I\sigma_{fg}^2)^{-1}\eta}
\,\,\Big) \nonumber \\
&=1-  \Phi \left( \Phi^{-1}(1-\alpha) - \sqrt{\eta^T(\Sigma_n+I\sigma_{fg}^2)^{-1}\eta} \,\right)\\
&= 1-  \Phi \left( \Phi^{-1}(1-\alpha) - \sqrt{\eta^T(\Sigma_n+I\lambda_{max}(\Sigma_n)\sigma_g^2)^{-1}\eta} \right).\nonumber
\end{align}
Likewise
$
\pi(R_g) 
=1-\Phi \left( \Phi^{-1}(1-\alpha) - \sqrt{\frac{\eta^T \Sigma_n^{-1}\eta}{\sigma_g^2+1}} \right). 
$
Therefore, $\pi(R_g) \geq \pi(R_{fg})$ if and only if 
\begin{equation*}\label{eq:power_inequality}
\frac{\eta^T \Sigma_n^{-1}\eta}{\sigma_g^2+1} \geq \eta^T(\Sigma_n+I\lambda_{max}(\Sigma_n)\sigma_g^2)^{-1}\eta.
\end{equation*}
Diagonalizing the two matrices $(1+\sigma^2_g)\Sigma_n$ and $\Sigma_n+I\lambda_{max}(\Sigma_n)\sigma_g^2$ by the common orthogonal matrix of their eigenvectors, we see that the diagonal terms, that is, the eigenvalues $(\sigma^2_g+1)\lambda_i$ of the first matrix are   less than or equal to those of the second one,  $\lambda_i+\lambda_{max}\sigma_g^2$. It follows that $\frac{\Sigma_n^{-1}}{\sigma^2_g+1}\succeq (\Sigma_n+I\lambda_{max}(\Sigma_n)\sigma_g^2)^{-1}$, where $A \succeq B$ means that $A-B$ is nonnegative definite; see \cite{horn2012matrix}, Chapter 4. The result follows.
\qed

\newpage

\end{document}